%% file: peleg.tex
\documentclass[proceedings]{stacs}
\stacsheading{2010}{609-620}{Nancy, France}
\firstpageno{609}

\theoremstyle{plain}
\theoremstyle{definition}

\usepackage{latexsym}
\usepackage{amsmath}
\usepackage{amssymb}

\newcommand{\eps}{\epsilon}

\newcommand{\ALGORITHM}{{\tt Algorithm}}
\newcommand{\SPAN}{\mbox{\bf disk-spanner}}

\newcommand{\NN}{\mbox{\sf NN}}

\def\SP{\mbox{\tiny SP}}

\def\DIR{\mbox{\tiny DIR}}

\begin{document}

\title[Relaxed spanners for directed disk graphs]{Relaxed spanners for directed disk graphs}

\author[lab1]{D. Peleg}{David Peleg}
\address[lab1]{Department of Computer Science and Applied Mathematics,
  \newline The Weizmann Institute of Science, Rehovot 76100, Israel}  
\email{david.peleg@weizmann.ac.il}  

\author[lab2]{L. Roditty}{Liam Roditty}
\address[lab2]{Department of Computer Science, Bar-Ilan University,
\newline Ramat-Gan 52900, Israel}   
\email{liamr@macs.biu.ac.il}  

\thanks{Thanks: Supported in part by grants from
the Minerva Foundation and the Israel Ministry of Science} 



\begin{abstract}
Let $(V,\delta)$ be a finite metric space, where $V$ is a set of
$n$ points and $\delta$ is a distance function defined for these
points. Assume that $(V,\delta)$ has a constant doubling dimension
$d$ and assume that each point $p\in V$ has a disk of radius
$r(p)$ around it. The disk graph that corresponds to $V$ and
$r(\cdot)$ is a \emph{directed} graph $I(V,E,r)$, whose vertices
are the points of $V$ and whose edge set includes a directed edge
from $p$ to $q$ if $\delta(p,q)\leq r(p)$. In~\cite{PeRo08} we
presented an algorithm for constructing a $(1+\eps)$-spanner of
size $O(n/\eps^d \log M)$, where $M$ is the maximal radius $r(p)$.
The current paper presents two results. The first shows that the
spanner of~\cite{PeRo08} is essentially optimal, i.e., for metrics
of constant doubling dimension it is not possible to guarantee a
spanner whose size is independent of $M$. The second result shows
that by slightly relaxing the requirements and allowing a small
perturbation of the radius assignment, considerably better
spanners can be constructed. In particular, we show that if it is
allowed to use edges of the disk graph $I(V,E,r_{1+\eps})$, where
$r_{1+\eps}(p) = (1+\eps)\cdot r(p)$ for every $p\in V$, then it
is possible to get a $(1+\eps)$-spanner of size $O(n/\eps^d)$ for
$I(V,E,r)$. Our algorithm is simple and can be implemented
efficiently.
\end{abstract}

\maketitle

\section*{Introduction}\label{S:one}

This paper concerns efficient constructions of spanners for disk
graphs, an important family of directed graphs. A {\em spanner} is
essentially a skeleton of the graph, namely, a sparse spanning
subgraph that faithfully represents distances. Formally, a
subgraph $H$ of a graph $G$ is a $t$-spanner of $G$ if
$\delta_H(u,v) \leq t \cdot\delta_G(u,v)$ for every two nodes $u$
and $v$, where $\delta_{G'}(u,v)$ denotes the distance between $u$
and $v$ in $G'$. We refer to $t$ as the {\em stretch factor} of
the spanner. Graph spanners have received considerable attention
over the last two decades, and were used implicitly or explicitly
as key ingredients of various distributed applications. It is
known how to efficiently construct a $(2k-1)$-spanner of size
$O(n^{1+1/k})$ for every weighted undirected graph, and this
size-stretch tradeoff is conjectured to be tight. Baswana and
Sen~\cite{BaSe07} presented a linear time randomized algorithm for
computing such a spanner. In directed graphs, however, the
situation is different. No such general size-stretch tradeoff can
exist, as indicated by considering the example of a directed
bipartite graph $G$ in which all the edges are directed from one
side to the other; clearly, the only spanner of $G$ is $G$ itself,
as any spanner for $G$ must contain every edge.

The main difference between undirected and directed graphs is that
in undirected graphs the distances are symmetric, that is, a path
of a certain length from $u$ to $v$ can be used also from $v$ to
$u$. In directed graphs, however, the existence of a path from $u$
to $v$ does not imply anything on the distance in the opposite
direction from $v$ to $u$. Hence, in order to obtain a spanner for
a directed graph one must impose some restriction either on the
graph or on its distances. In order to bypass the problem of
asymmetric distances of directed graphs, Cowen and
Wagner~\cite{CoWa99} introduced the notion of {\em roundtrip
distances} in which the distance between $u$ and $v$ is composed
of the shortest path from $u$ to $v$ plus the shortest path from
$v$ to $u$. It is easy to see that under this definition distances
are symmetric also in directed graphs. It is shown by Cowen and
Wagner~\cite{CoWa99} and later by Roditty, Thorup and
Zwick~\cite{RoThZw08} that methods of path approximations from
undirected graphs can work using more ideas also in directed
graphs when roundtrip distances are considered. Bollob{\'a}s,
Coppersmith and Elkin~\cite{BoCoEl05} introduced the notion of
{\em distance preservers} and showed that they exist also in
directed graphs.

In~\cite{PeRo08} we presented a spanner construction for directed
graphs without symmetric distances. The restriction that we
imposed on the graph was that it must be a disk graph. More
formally, let $(V,\delta)$ be a finite metric space of constant
doubling dimension $d$, where $V$ is a set of $n$ points and
$\delta$ is a distance function defined for these points. A metric
is said to be of {\em constant doubling dimension} if a ball with
radius $r$ can be covered by at most a constant number of balls of
radius $r/2$. Every point $p\in V$ is assigned with a radius
$r(p)$. The disk graph that corresponds to $V$ and $r(\cdot)$ is a
directed graph $I(V,E,r)$, whose vertices are the points of $V$
and whose edge set includes a directed edge from $p$ to $q$ if $q$
is inside the disk of $p$, that is, $\delta(p,q)\leq r(p)$.
In~\cite{PeRo08} we presented an algorithm for constructing a
$(1+\eps)$-spanner with size $O(n/\eps^d \log M)$, where $M$ is
the maximal radius. In the case that we remove the radius
restriction the resulted graph is the complete undirected graph
where the weight of every edge is the distance between its
endpoint. In such a case it is possible to create
$(1+\eps)$-spanners of size $O(n/\eps^d)$,
see~\cite{HaMe06},~\cite{GaoGuiNgu04} and \cite{Ro07b} for more
details. Moreover, when the radii are all the same and the graph
is the unit disk graph then it is also possible to create
$(1+\eps)$-spanners of size $O(n/\eps^d)$,
see~\cite{gao05geometric}, \cite{PeRo08}.

As a result of that, a natural question is whether a spanner size
of $O(n/\eps^d \log M)$ in the case of directed disk graph is
indeed the best possible or maybe it is possible to get a spanner
of size $O(n/\eps^d)$ as in the cases of the complete graph and
the unit disk graph. For the case of the Euclidean metric space,
the answer turns out to be positive; a simple modification of the
Yao graph construction~\cite{Yao82} to fit the directed case
yields a directed spanner of size $O(n/\eps^d)$. However, the
question remains for more general metric spaces, and in particular
for the important family of metric spaces of bounded doubling
dimension.

In this paper we provide an answer for this question. We show that
our construction from~\cite{PeRo08} is essentially optimal by
providing a metric space with a constant doubling dimension and a
radius assignment whose corresponding disk graph has $\Omega(n^2)$
edges and none of its edges can be removed. (This does not
contradict our spanner construction from~\cite{PeRo08} as the
maximal radius in that case is $\Theta(2^n)$ and hence $\log M =
n$.)

This (essentially negative) optimality result motivates our main
interest in the current paper, which focuses on attempts to
slightly relax the assumptions of the model, in order to obtain
sparser spanner constructions. Indeed, it turns out that such
sparser spanner constructions are feasible under a suitably
relaxed model. Specifically, we demonstrate the fact that if a
small perturbation of the radius assignment is allowed, then a
$(1+\eps)$-spanner of size $O(n/\eps^d)$ is attainable. More
formally, we show that if we are allowed to use edges of the disk
graph $I(V,E,r_{1+\eps})$, where $r_{1+\eps}(p) = (1+\eps)\cdot
r(p)$ for every $p\in V$, then it is possible to get a
$(1+\eps)$-spanner of size $O(n/\eps^d)$ for the original disk
graph $I(V,E,r)$. This approach is similar in its nature to the
notation of \textit{emulators} introduced by Dor, Halperin and
Zwick~\cite{DoHaZw00}. An emulator of a graph may use any edge
that does not exist in the graph in order to approximate its
distances. It was used in the context of spanners with an additive
stretch.

The main application of disk graph spanners is for topology
control in the wireless ad hoc network model. In this model the
power required for transmitting from $p$ to $q$ is commonly taken
to be $\delta(p,q)^\alpha$, where $\delta(p,q)$ denotes the
distance between $p$ and $q$ and $\alpha$ is a constant typically
assumed to be between $2$ and $4$. Most of the ad hoc network
literature makes the assumption that the transmission range of all
nodes is identical, and consequently represents the network by a
\emph{unit disk graph} (UDG), namely, a graph in which two nodes
$p,q$ are adjacent if their distance satisfies $\delta(p,q) \le
1$. A unit disk graph can have as many as $O(n^2)$ edges.

There is an extensive body of literature on spanners of unit disk
graphs. Gao et al.~\cite{gao05geometric}, Wang and
Yang-Li~\cite{WaLi06a} and Yang-Li et al.~\cite{Yang2003}
considered the restricted Delaunay graph, whose worst-case stretch
is constant (larger than $1+\epsilon$). In~\cite{PeRo08} we showed
that any $(1+\epsilon)$-geometric spanner can be turned into a
$(1+\epsilon)$-UDG spanner.

Disk graphs are a natural generalization of unit disk graphs, that
provide an intermediate model between the complete graph and the
unit disk graph. Our size efficient spanner construction for disk
graphs whose radii are allowed to be slightly larger falls exactly
into the model of networks in which the stations can change their
transmission power. In particular our constriction implies that if
any station increases its transmission power by a small fraction
then a considerably improved topology can be built for the
network.

Our result has both practical and theoretical implications. From a
practical point of view it shows that, in certain scenarios,
extending the transmission radii even by a small factor can
significantly improve the overall quality of the network topology.
The result is also very intriguing from a theoretical standpoint,
as to the best of our knowledge, our relaxed spanner is the first
example of a spanner construction for directed graphs that enjoys
the same properties as the best constructions for undirected
graphs. (As mentioned above, it is easy to see that for general
directed graphs, it is not possible to have an algorithm that
given any directed graph produces a sparse spanner for it.) In
that sense, our result can be viewed as a significant step towards
gaining a better understanding for some of the fundamental
differences between directed and undirected graphs. Our result
also opens several new research directions in the relaxed model of
disk graphs. The most obvious research questions that arise are
whether it is possible to obtain other objects that are known to
exist in undirected graphs, such as compact routing schemes and
distance oracles, for disk graphs as well.

The rest of this paper is organized as follows. In the next
section
we present a metric space of constant doubling dimension in which
no edge can be removed from its corresponding disk graph. Section
\ref{s:DG-3} first describes a simple variant of our construction
from~\cite{PeRo08}, and then uses it together with new ideas in
order to obtain our new relaxed construction. Finally, in Section
\ref{s:con} we present some concluding remarks and open problems.

\section{Optimality of the spanner construction}
\label{s:Bad-Example}

In this section we build a disk graph $G$ with $2n$ vertices and
$\Omega(n^2)$ edges that is non-sparsifiable, namely, whose only
spanner is $G$ itself. In this graph $M=\Omega(2^n)$ hence our
spanner construction from~\cite{PeRo08} has a size of
$\Omega(n^2)$ and is essentially optimal.

Given a set of points, we present a distance function such that
for a given assignment of radii for the points any spanner of the
resulting disk graph must have $\Omega(n^2)$ edges. We then prove
that the underlying metric space has a constant doubling
dimension.

We partition the points into two types, $Y=\{y_1, \ldots, y_n\}$
and $X=\{x_1,\ldots, x_n\}$. We now define the distance function
$\delta(\cdot,\cdot)$ and the radii assignment $r(\cdot)$. The
main idea is to create a bipartite graph $G(X,Y,E)$ in which every
point of $Y$ is connected by a directed edge to all the points of
$X$.

The distance between any two points $x_i$ and $x_j$ is at least
$1+\eps$ for some small $0<\eps<1$ and the radius assignment of
every point $x_i$ is exactly $1$. Thus, there are no edges between
the points of $X$.

We now define the distances between the points of $Y$ and the
points of $X$. We start with the point $y_1$. Let
$\delta(y_1,x_i)=n$ for every $x_i \in X$ and let $r(y_1) = n$.
Place the points of $X$ on the boundary of a ball of radius $n$
centered at $y_1$ such that the distance between any two
consecutive points $x_i$ and $x_{i+1}$ is exactly $1+\eps$. This
is depicted in Figure~\ref{F-badexample}(a).

\begin{figure}[!t]
\begin{center}
\input{example-combined-tiny_m.pstex_t}
\end{center}
\caption{(a) First step in constructing the non-sparsifiable disk
graph $G$. (b) The non-sparsifiable disk graph $G$.}
\label{F-badexample}
\end{figure}
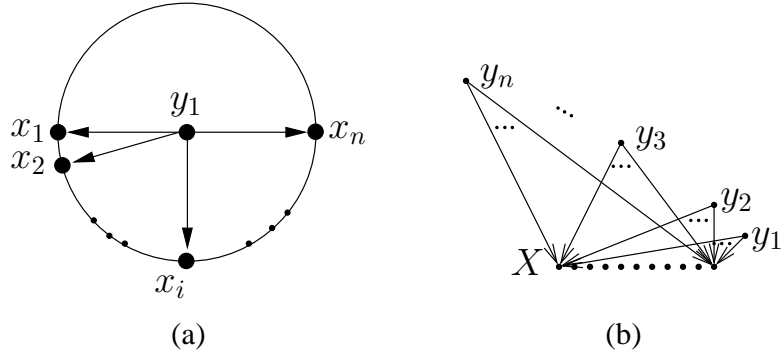

Turning to the point $y_2$, let $\delta(y_2,x_i)=2n$ for every
$x_i \in X$, $\delta(y_2,y_1)=2n+\eps$, and $r(y_2) = 2n$. Hence
there is an edge from $y_2$ to all the points of $X$, but no edge
connects $y_2$ and $y_1$.

We now turn to define the general case. Consider $y_i \in Y$. Let
$r(y_i)=2^{i-1}n$ and $\delta(y_i,x_j)=2^{i-1}n$ for every $x_j\in
X$. Let $\delta(y_i,y_{i-1})=2^{i-1}n+\eps$, and in general, for
every $0<j<i$ we have
\begin{equation}
\label{eq:dist} \delta(y_i,y_j) ~=~
\sum_{k=j}^{i-1}\delta(y_{k+1},y_k)~,
\end{equation}
implying that
\begin{equation}
\label{eq:dist2} \delta(y_i,y_j) ~<~ 2^in.
\end{equation}
It is easy to verify that $y_i$ has outgoing edges to the points
of $X$ (and to them only) and it does not have any incoming edges.
See Figure \ref{F-badexample}(b).

The resulting disk graph $G$ has $2n$ vertices and $\Omega(n^2)$
edges. Clearly, removing any edge from $G$ will increase the
distance between its head and its tail to infinity, and thus the
only spanner of $G$ is $G$ itself.

It is left to show that the metric space defined above for $G$ has
a constant doubling dimension. Given a metric space $(V,\delta)$,
its {\em doubling dimension} is defined to be the minimal value
$d$ such that every ball $B$ of radius $r$ in the metric space can
be covered by $2^{d}$ balls of radius $r/2$. In the next Theorem
we prove that for the metric space described above, $d$ is
constant.

\begin{theorem}
The metric space $(X\cup Y,\delta)$ defined for $G$ has a constant
doubling dimension.
\end{theorem}
\begin{proof}
Let $B$ be a ball with an arbitrary radius $r$. We show that it is
possible to cover all the points of $X\cup Y$ within $B$ using a
constant number of balls whose radius is $r/2$. The proof is
divided into two cases.

{\bf Case a:} There is some $y_j \in Y$ within the ball $B$. (If
there is more than one such point, then let $y_j$ be the point
whose index is maximal.) Let $B'$ be a ball of radius $R=2r$
centered at $y_j$. Clearly $B\subset B'$, so $B'$ contains all the
points of $B$. In what follows we show that all the points of
$X\cup Y$ within $B'$ can be covered by a constant number of balls
of radius $r/2$. Let $y_i$ be the point within $B'$ whose index is
maximal.  We have to consider two possible scenarios. The first is
that $y_j=y_i$. This implies that $y_{j+1} \notin B'$, hence $R <
\delta(y_{j+1},y_j)=2^{j}n + \eps$. We now show that it is
possible to cover $B'$ by a constant number of balls of radius
$R/4$. If $R<2^{j-1}n$, then only $y_j$ is within $B'$ and it is
covered by a ball of radius $R/4$ centered at itself. If $2^{j-1}n
\leq R<2^{j-1}n+\eps$, then $B'$ contains all the points of $X$
and $y_j$. From packing arguments it follows that it is possible
to cover all the points of $X$ by a constant number of balls of
radius $n/4$, hence also by a constant number of balls of radius
$R\geq n$. The point $y_j$ itself is covered by a ball centered at
it. Finally, if $2^{j-1}n+\eps \leq R<2^{j}n+\eps$, then $R/4$ is
at least $2^{j-3}n+\eps/4$.
A ball centered at $y_{j-3}$ of radius $R/4$ covers every $y_k$
within $B'$, where $1\le k\leq j-3$, as $\delta(y_{j-3},y_k) \leq
2^{j-3}n$. Hence, we cover $Y\cap B'$ by balls of radius $R/4$
whose centers are $y_j$, $y_{j-1}$, $y_{j-2}$ and $y_{j-3}$. We
cover $X\cap B'$ as before. This completes the first scenario,
where $y_i = y_j$. Assume now that $y_i\neq y_j$. This implies
that $\delta(y_i,y_j)\leq R$ and that $R < \delta(y_{i+1},y_{j})$,
where the first inequality follows from the fact that $y_i\in B'$
and the second inequality follows from the fact that $y_i$ is the
point with maximal index inside $B'$, hence, $y_{i+1}\notin B'$.
As $\delta(y_{i},y_{i-1})\leq\delta(y_i,y_j)$, we get that
$2^{i-1}n+\eps \leq R$. Also, by (\ref{eq:dist2}),
$\delta(y_{i+1},y_{j})<2^{i+1}n$. We conclude that $2^{i-1}n \leq
R < 2^{i+1}n$ and that $R/4 \geq  2^{i-3}n$. A ball centered at
$y_{i-3}$ of radius $R/4$ covers every $y_k$ within $B'$, where
$k\leq i-3$, as $\delta(y_{i-3},y_k) \leq 2^{i-3}n$. Hence, we can
cover $B'\cap Y$ by balls of radius $R/4$ whose centers are $y_i$,
$y_{i-1}$, $y_{i-2}$ and $y_{i-3}$. We cover $X\cap B'$ as before.
This completes the first case.

{\bf Case b:} The ball $B$ does not contain any point from $Y$.
The points of $X$ are spread as appears in
Figure~\ref{F-badexample}(a), thus by standard packing arguments,
any ball that contains only points from $X$ is covered by a
constant number of balls of half the radius.
\end{proof}

\section{Improved spanner in the relaxed disk graph model}
\label{s:DG-3}

The (negative) optimality result from the previous section
motivates us to look for a slightly relaxed definition of disk
graphs in which it will still be possible to create a spanner of
size $O(n/\eps^d)$.

Let $(V,\delta)$ be a metric space of constant doubling dimension
$d$ with a radius assignment $r(\cdot)$ for its points and let
$I=(V,E,r)$ be its corresponding disk graph. Assume that we
multiply the radius assignment of every point by a factor of
$1+\eps$, for some $\eps>0$, and let $I'=(V,E',r_{1+\eps})$ be the
corresponding disk graph. It is easy to see that $E\subseteq E'$.
In this section we show that it is possible to create a
$(1+\eps)$-spanner of size $O(n/\eps^d)$ if we are allowed to use
edges of $I'$. As a first step we present a simple variant of our
$(1+\eps)$-spanner construction of size $O(n/\eps^d\log M)$
from~\cite{PeRo08}. This variation is needed in order to obtain
the efficient construction in the relaxed model which is presented
right afterwards.

\subsection{Spanners for general disk graphs} \label{s:DG}

Let $(V,\delta)$ be a metric space of constant doubling dimension
and assume that any point $p\in V$ is the center of a ball of
radius $r(p)$, where $r(p)$ is taken from the range $[1,M]$. In
this section we describe a simple variant of our construction
from~\cite{PeRo08}, which computes a $(1+\epsilon)$-spanner with
$O(n/\epsilon^{d}\log M)$ edges for a given disk graph. We then
use this variant, together with new ideas, in order to obtain (in
the next section) our main result, namely, a spanner with only
$O(n/\epsilon^{d})$ edges.

The spanner construction algorithm receives as input a directed
graph $I(V,E,r)$ and an arbitrarily small (constant) approximation
factor $\epsilon>0$, and constructs a set of spanner edges
$E_{\SP}^{\DIR}$, returning the spanner subgraph
$H^{\DIR}(V,E_{\SP}^{\DIR})$. The construction of the spanner is
based on a hierarchical partition of the points of $V$ that takes
into account the different radius of each point. The construction
operates as follows. Let $\alpha$ and $\beta$ be two small
constants depending on $\epsilon$, to be fixed later on. Assume
that the ball radii are scaled so that the smallest edge in the
disk graph is of weight $1$. Let $i$ be an integer from the range
$[0,\lfloor \log_{1+\alpha}M\rfloor]$ and let $M_i=M /
(1+\alpha)^i$. The edges of $I(V,E,r)$ are partitioned into
classes by length, letting $E(M_{i+1},M_i) = \{ (x,y) \mid M_{i+1}
\leq \delta(x,y) \leq M_i \}$. Let $\ell(x,y)$ be the level of the
edge $(x,y)$, that is, $\ell(x,y)=i$ such that $(x,y) \in
E(M_{i+1},M_i)$. Let $p$ be a point whose ball is of radius
$r(p)\in [M_{i+1},M_i]$. It follows that level $i$ is the first
level in which $p$ can have outgoing edges. We denote this level
by $\ell(p)$.

For every $i \in [0,\lfloor \log_{1+\alpha}M\rfloor]$, starting
from $i=0$, the edges of the class $E(M_{i+1},M_i)$ are considered
by the algorithm in a non-decreasing order. (Assume that in each
class the edges are sorted by their weight.) In each stage of the
construction we maintain a set of pivots $P_i$. Let $x\in V$ and
let $\NN(x,P_i)$ be the nearest neighbor of $x$ among the points
of $P_i$. For a pivot $p\in P_i$, define $\Gamma_i(p)= \{ x \mid x
\in V, \NN(x,P_i) = p, r(x)\geq \delta(x,p) \}$, namely, all the
points that have a directed edge to $p$ and $p$ is their nearest
neighbor from $P_i$. We refer to $\Gamma_i(p)$ as the {\em close
neighborhood} of $p$.

The algorithm is given in Figure~\ref{F-App-Alg}. Let $(x,y)$ be
an edge considered by the algorithm in the $i$th iteration. The
algorithm first checks whether $x$ or $y$ or both should be added
to the pivots set $P_i$. The main change with respect
to~\cite{PeRo08} is that if $y$ is assigned with a large enough
radius it might become a pivot when the edge $(x,y)$ is examined.
When considering the edge $(x,y)$, the algorithm acts according to
the following rule: If the distance from $x$ to its nearest
neighbor in $P_i$ is greater than $\beta M_{i+1}$ then $x$ is
added to $P_i$. If the distance from $y$ to its nearest neighbor
in $P_i$ is greater than $\beta M_{i+1}$ and the radius of $y$ is
at least $M_{i+1}$ then $y$ is added to $P_i$. To decide whether
the edge $(x,y)$ is added to the spanner, the following two cases
are considered. The first case is when $r(y) \geq M_{i+1}$. In
this case, if there is no edge from the close neighborhood of $x$
to the close neighborhood of $y$ then $(x,y)$ is added to the
spanner. The second case is when $r(y) < M_{i+1}$. In this case,
if there is no edge from the close neighborhood of $x$ to $y$ then
$(x,y)$ is added to the spanner. When $i$ reaches $\lfloor
\log_{1+\alpha}M\rfloor$, the algorithm handles all the edges that
belong to $E(M_{\lfloor \log_{1+\alpha}M\rfloor+1},M_{\lfloor
\log_{1+\alpha}M\rfloor})$. This includes also edges whose weight
is $1$, the minimal possible weight. The algorithm returns the
directed graph $H^{\DIR}(V,E_{\SP}^{\DIR})$.

In what follows we prove that for suitably chosen $\alpha$ and
$\beta$, $H^{\DIR}(V,E_{\SP}^{\DIR})$ is a $(1+\epsilon)$-spanner
with $O(n/\epsilon^{d}\log M)$ edges of the directed graph
$I(V,E,r)$.

\begin{figure}[t]
\begin{center}
\framebox{\parbox[t]{12.0in}{
\begin{tabbing}
\ALGORITHM$\;$ \SPAN$$ $(I(V,E,R),\epsilon)$\\[5pt]
$E_{\SP}^{\DIR} \gets \phi$\\
$P_0 \gets \phi$\\
for $\;\;$\=$i\gets 0$ to $\lfloor\log_{1+\alpha}M\rfloor$\\
\> for $\;$\= each $(x,y) \in E(M_{i+1},M_i)$ do\\
\>\> if\=\ $\delta(\NN(x,P_i),x)> \beta M_{i+1}$ then
$P_i \gets P_i \cup \{x\}$ \\
\>\> if\=\ $\delta(\NN(y,P_i),y)> \beta M_{i+1}\wedge r(y)\geq
M_{i+1}$ then
$P_i \gets P_i \cup \{y\}$ \\
\>\> if\=\ $r(y) \geq M_{i+1}$\\
\>\>\> if\=\ $\nexists (x',y')\in E_{\SP}^{\DIR}$ s.t. \=$x'\in
\Gamma_i(\NN(x,P_i))\wedge$
$y'\in\Gamma_i(\NN(y,P_i))$ \\
\>\>\>\> then $E_{\SP}^{\DIR} \gets E_{\SP}^{\DIR} \cup \{(x,y)\}$\\
\>\> if\=\ $r(y) < M_{i+1}$\\
\>\>\> if\=\ $\nexists (x',y)\in E_{\SP}^{\DIR}$
s.t. $x'\in \Gamma_i(\NN(x,P_i))$ \\
\>\>\>\> then $E_{\SP}^{\DIR} \gets E_{\SP}^{\DIR} \cup \{(x,y)\}$\\
\> $P_{i+1} \gets P_i$\\
return $H^{\DIR}(V,E_{\SP}^{\DIR})$
\end{tabbing}
}}
\end{center}
\caption{ A high level implementation of the spanner construction
algorithm for \emph{general} disk graphs }\label{F-App-Alg}
\end{figure}

\begin{lemma}[Stretch]
Let $\epsilon > 0$, set $\alpha = \beta < \epsilon/6$ and let
$H=H^{\DIR}(V,E_{\SP}^{\DIR})$ be the graph returned by Algorithm
$\SPAN(I(V,E,r),\epsilon)$. If $(x,y)\in E$ then $\delta_H(x,y)
\leq (1+\epsilon)\delta(x,y)$.
\end{lemma}
\begin{proof}
Recall that the radii are scaled so that the shortest edge is of
weight $1$. We prove that every directed edge of an arbitrary node
$x\in V$ is approximated with $1+\epsilon$ stretch. Let $i\in
[0,\lfloor\log_{1+\alpha}M\rfloor$]. The proof is by induction on
$i$. For a given node $x$, the base of the induction is the
maximal value of $i$ in which $x$ has an edge in $E(M_{i+1},M_i)$.
Let $j$ be this value for $x$, that is, the set $E(M_{j+1},M_j)$
contains the shortest edge that touches $x$. Every other node is
at distance at least $M_{j+1}$ away from $x$, hence $x$ is a pivot
at this stage and every edge that touches $x$ from the set
$E(M_{j+1},M_j)$ is added to $E_{\SP}^{\DIR}$.

Let $(x,y) \in E(M_{i+1},M_i)$ for some $i<j$ and let
$p=\NN(x,P_i)$. Assume that $r(y)\geq M_{i+1}$ and let
$q=\NN(y,P_i)$. It follows from definition that $\delta(x,p) \leq
\beta M_{i+1}$ and $\delta(y,q) \leq \beta M_{i+1}$.

If the edge $(x,y)$ is not in the spanner, then there must be an
edge $(\hat{x},\hat{y})\in E_{\SP}^{\DIR}$, where $\hat{x}\in
\Gamma_i(p)$ and $\hat{y}\in \Gamma_i(q)$. The crucial observation
is that the radius of $x$ and $\hat{y}$ is at least $M_{i+1}$. By
the choice of $\beta$, it follows that $2\beta M_{i+1}< M_{i+1}$
and $(x,\hat{x}),(\hat{y},y)\in E$. Thus, there is a (directed)
path from $x$ to $y$ of the form $\langle
x,\hat{x},\hat{y},y\rangle$ whose length is $4\beta M_{i+1}+M_i$.
However, only its middle edge, $(\hat{x},\hat{y})$, is in
$E_{\SP}^{\DIR}$. The length of this edge is bounded by the length
of the edge $(x,y)$ since the algorithm picked the minimal edge
that connects between the neighborhoods. This implies that the
length of $(\hat{x},\hat{y})$ is at most $M_i$.

By the inductive hypothesis, the edges $(x,\hat{ x})$ and
$(\hat{y},y)$ whose weight is at most $2\beta M_{i+1}$ are
approximated with $1+\epsilon$ stretch. Thus, there is a path in
the spanner from $x$ to $y$ whose length is at most
$(1+\epsilon)\delta(x,\hat{x})+M_i+(1+\epsilon)\delta(\hat{y},y),$
and this can be bounded by
$$(1+\epsilon)4\beta M_{i+1}+M_i ~=~
((1+\epsilon)4\beta +(1+\alpha)) M_{i+1}.$$ As the edge $(x,y) \in
E(M_{i+1},M_i)$ it follows that $\delta(x,y)\geq M_{i+1}$. It
remains to prove that $1 + 4\epsilon\beta + 4\beta + \alpha \leq
1+\epsilon$, which follows directly from the choice of $\alpha$
and $\beta$.

If $r(y) < M_{i+1}$ then there must be an edge $(\hat{x},y)\in
E_{\SP}^{\DIR}$, where $\hat{x}\in \Gamma_i(p)$. Following similar
arguments to those used above it can be shown that there is a path
in the spanner from $x$ to $y$ of length at most
$(1+\epsilon)2\beta M_{i+1}+M_i$ and hence bounded by
$(1+\epsilon)M_{i+1}$.
\end{proof}

\paragraph{The size of the spanner.}
We now prove that the size of the spanner
$H^{\DIR}(V,E_{\SP}^{\DIR})$ is $O(n/\epsilon^d \log M)$. As a
first step, we state the following well-known lemma,
cf.~\cite{GaoGuiNgu04}.

\begin{lemma}\label{L-pack}[Packing Lemma]
If all points in a set $U \in \mathbb{R}^d$ are at least $r$ apart
from each other, then there are at most $(2R/r + 1)^d$ points in
$U$ within any ball $X$ of radius $R$.
\end{lemma}

The next lemma establishes a bound on the number of incoming
spanner edges that a point may be assigned on stage $i \in
[0,\lfloor\log_{1+\alpha}M\rfloor]$ of the algorithm.

\begin{lemma}\label{L-level-size}
Let $i \in [0,\lfloor\log_{1+\alpha}M\rfloor]$ and let $y\in V$.
The total number of incoming edges of $y$ that were added to the
spanner on stage $i$ is $O(\epsilon^{-d})$.
\end{lemma}
\begin{proof}
Let $(x,y)$ be a spanner edge and let $\NN(x,P_i) = p$. We
associate $(x,y)$ to $p$. From the spanner construction algorithm
it follows that this is the only incoming edge of $y$ whose source
is in $\Gamma_i(p)$. Thus, this is the only incoming edge of $y$
which is associated to $p$. Now consider all the incoming edges of
$y$ on stage $i$. The source of each of these edges is associated
to a unique pivot within distance of at most $M_i+2\beta M_{i+1}$
away from $y$ and any two pivots are $\beta M_{i+1}$ apart from
each other. Using Lemma~\ref{L-pack}, we get that the number of
edges entering $y$ is $(\frac{M_i+2\beta M_{i+1}}{\beta M_{i+1}} +
1)^d=((1+\alpha)/\beta+3)^d=O(\epsilon^{-d})$.
\end{proof}

It follows from the above lemma that the total number of edges
that were added to $E_{\SP}^{\DIR}$ in the main loop is
$O(n/\epsilon^d \log M)$. The total cost of the construction
algorithm is $O(m\log n)$. For more details on the construction
time see~\cite{PeRo08}.

\subsection{Spanner for relaxed disk graphs}

Let $(V,\delta)$ be a metric space of constant doubling dimension
$d$ with a radius assignment $r(\cdot)$ for its points and let
$I=(V,E,r)$ be its corresponding disk graph. Assume that we
multiply the radius assignment of every point by a factor of
$1+\eps$, for some $\eps>0$, and let $I'=(V,E',r_{1+\eps})$ be the
corresponding disk graph. In this section we show that it is
possible to create a $(1+\eps)$-spanner of $I$ of size
$O(n/\eps^d)$ if we are allowed to use edges of $I'$.

Our construction consists of two stages: a building stage and a
pruning stage. The building stage creates two spanners, $H$ and
$H'$, using the algorithm of Section~\ref{s:DG}, where $H$ is the
spanner of $I$ and $H'$ is the spanner of $I'$. In the pruning
stage we prune the union of these two spanners. Throughout the
pruning stage we use the radius assignment of each point before
the increase. Let $q\in V$ and let $\ell(q)$ be the first level in
which $q$ can have outgoing edges, that is, $r(q)\in
[M_{\ell(q)+1},M_{\ell(q)}]$ (recall that as the levels get larger
the edges get shorter). In the pruning stage we only prune
incoming edges of $q$ whose level is below $\ell(q)$. In other
words, we do not touch the incoming edges of $q$ that are shorter
than the radius of $q$. The pruning is done as follow. Let $\gamma
= \log_{1+\alpha}{1/\beta}+1$. We keep in the spanner the incoming
edges of $q$ that come from the first $4\gamma$ different levels
below $\ell(q)$.

Let $\hat{H}$ be the resulting spanner and let $\hat{E}$ be the
remaining set of edges after the pruning step. In the remainder of
this section we show that the size of $\hat{H}$ is
$O(n/\epsilon^d)$ and its stretch with respect to the distances in
$I(V,E,r)$ is $1+\eps$. We start by showing that the size of
$\hat{H}$ is $O(n/\epsilon^d)$. Notice that the first part of the
proof below is possible only due to the change we have done in the
previous section to our spanner construction from~\cite{PeRo08}.
Roughly speaking, given an edge $(p,q) \in E$ that is shorter than
$r(q)$ we use pivot selection also on $q$'s side (and not only on
$p$'s) to sparisify the graph. This allows us to deal separately
with edges of $q$ of length larger than $r(q)$ and those of length
smaller than $r(q)$.

\begin{lemma}\label{L-smaller}
$|\hat{E}| = O(n/\epsilon^d)$.
\end{lemma}
\begin{proof}
Let $(p,q)$ be a spanner edge that survived the pruning step.
There are two possible cases to consider.

The first case is that $\ell(p,q) > \ell(q)$. Let $i=\ell(p,q)$
and let $x=\NN(p,P_i)$ and $y=\NN(q,P_i)$. By packing
considerations similar to Lemma~\ref{L-level-size} it follows that
the total number of edges at level $i$ that connects between two
pivots as the edge $(p,q)$ that are associated with $x$ (and with
$y$) is $O(1/\epsilon^d)$. The distance between $x$ and $y$ is at
most $2\beta M_{i+1}+M_i$, therefore at level $i-2\gamma$ either
$x$ or $y$ are no longer pivots.

Let $x\in P_j$ and $x\not\in P_{j-1}$, that is, $P_j$ is the first
pivot set that contains $x$. Then we charge $x$ with every
(incoming and outgoing) edge of this type from levels
$[j,j+2\gamma]$ that is incident to $x$. Now given such an edge
$(p,q)$ whose level is $i$, either $x$ or $y$ are not pivots in
level $i-2\gamma$, which means that either $x$ or $y$ has been
charged for this edge, since one of them first becomes a pivot
between levels $i-2\gamma$ and $i$.

The second case is that $\ell(p,q) \leq \ell(q)$. In this case, it
must be that level $\ell(p,q)$  is among the $4\gamma$ first
different levels below $\ell(q)$ from which an incoming edge is
allowed to enter $q$. Subsequently, we associate the edge $(p,q)$
with $q$, as the total number of such edges that $q$ can have is
$O(\gamma / \epsilon^d)$.
\end{proof}

We now turn to prove that the stretch of the spanner $\hat{H}$
with respect to the disk graph $I$ is $1+\eps$.

\begin{lemma}\label{L-Stretch1}
Let $(p,q)$ be an edge of the spanner $H$ that was pruned. We show
that there is a path in $\hat{H}$ whose length is at most
$(1+\eps)\delta(p,q)$.
\end{lemma}
\begin{proof}
The proof is by induction on the lengths of the pruned edges. For
the induction base let $(p,q)$ be the shortest edge that was
pruned. For every $x\in V$, let $s(x)$ be the head of an edge
whose level is the $\gamma$-th level below $\ell(x)$ from which
$x$ has an incoming edge. Let $q_1,\ldots q_i, \ldots$ be a
sequence of points, where  $q_1 = q$ and $q_i = s(q_{i-1})$. As
$q_{i+1} = s(q_{i})$, it follows that $\ell(q_{i+1},q_{i})\leq
\ell(q_{i})-\gamma$. Combining this with the fact that
$\ell(q_i)\leq \ell(q_i,q_{i-1})$ we get that
$\ell(q_{i+1},q_i)\leq \ell(q_i,q_{i-1})-\gamma$. Therefore,
$\delta(q_i,q_{i-1})\leq \beta \delta(q_{i+1},q_i)$.

The analysis distinguishes between two cases.

{\bf Case a:} There is a point $q_t$ such that $\delta(q_t,q)> \beta
\delta(p,q)$. This situation is depicted in
Figure~\ref{F-proofexample}. (If there is more than one point that
satisfies this requirement, take the one whose index is minimal.)

\par\bigskip\noindent {\bf Claim:}
$\delta(q_t,q_{t-1})\geq \frac{\beta}{2} \delta(p,q)$.

\begin{proof}
For the sake of contradiction, assume that
$\delta(q_t,q_{t-1})<\frac{\beta}{2}\delta(p,q)$. This implies
that
\begin{equation}
\label{eq:qtqt-1} 2\delta(q_t,q_{t-1}) ~<~ \beta\delta(p,q) ~<~
\delta(q_t,q) ~\leq~ \sum_{i=2}^{t} \delta(q_{i}, q_{i-1})~,
\end{equation}
where the last inequality follows from the triangle inequality as
the distance between $q$ and $q_t$ is at most $\sum_{i=2}^{t}
\delta(q_{i-1}, q_{i})$.
For every $2\leq i\leq t-1$ we have $\delta(q_i,q_{i-1}) \leq
\beta \delta(q_{i+1},q_i)$, which implies that
$\delta(q_i,q_{i-1}) \leq \beta^{t-i} \delta(q_t , q_{t-1})$.
Combined with (\ref{eq:qtqt-1}), we get
$$\delta(q_t,q_{t-1}) ~<~ \sum_{i=2}^{t-1} \delta(q_i,q_{i-1}) ~\leq~
\delta(q_t ,q_{t-1}) \sum_{i=2}^{t-1}  \beta^{t-i}~.$$ If $\beta <
1/2$ we have $\sum_{i=2}^{t-1}  \beta^{t-i}<1$ and this yields a
contradiction.
\end{proof}

We now focus our attention on the point $q_{t-1}$.  The minimality
of $q_t$ implies that $\delta(q,q_{t-1})\leq \beta \delta(p,q)$.
By combining it with the triangle inequality we get that
$\delta(p,q_{t-1})\leq \delta(p,q)+\beta\delta(p,q)$. Therefore,
in the graph $I'$ there must be an edge from $p$ to $q_{t-1}$.

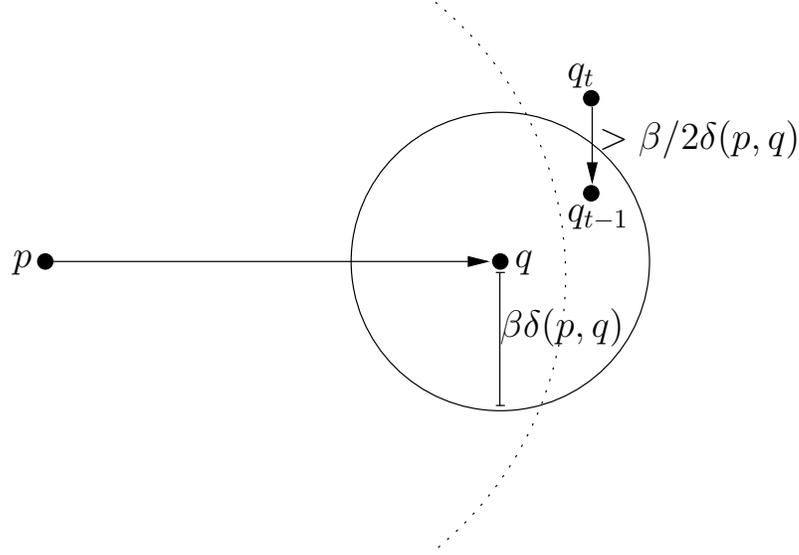
\begin{figure}[!t]
\begin{center}
\input{proof-small_m.pstex_t}
\end{center}
\caption{The case in which $q_t$ exists} \label{F-proofexample}
\end{figure}

Let $i=\ell(p,q_{t-1})$. There are two possible scenarios for the
spanner $H'$. The first scenario is when $r'(q_{t-1}) < M_{i+1}$.
In this case, there is an edge in $H'$ from some $x \in
\Gamma_{i}(\NN(p,i))$ to $q_{t-1}$, whose length is at most
$\delta(p,q)+\beta\delta(p,q)$.

There are $4\gamma$ different levels below $\ell(q_{t-1})$ from
which edges that belong to the spanners $H$ and $H'$ are not being
pruned and survived to the spanner $\hat{H}$. We know that the
edge $(q_t,q_{t-1})$ is such an edge from the $\gamma$-th
non-empty level below $\ell(q_{t-1})$. We also know that
$\delta(q_t,q_{t-1})>\frac{\beta}{2} \delta(p,q)$. Therefore, as
the length of the edge $(x,q_{t-1})$ is at most
$\delta(p,q)+\beta\delta(p,q)$ it is within the $4\gamma$
non-empty levels below $\ell(q_{t-1})$ and it is not pruned. We
can now build a path from $p$ to $q$ by concatenating three
segments as follows: A path from $p$ to $x$, the edge
$(x,q_{t-1})$ and a path from $q_{t-1}$ to $q$. The point $x$ is
at most $2\beta\delta(p,q)+2\beta^2\delta(p,q)$ away from $p$ and
for the right choice of $\beta$ it is less than
$\delta(p,q)/(1+\eps)$, hence the weight of every edge on the path
that approximates the distance between $x$ and $p$ in $H\cup H'$
is less than $\delta(p,q)$, the shortest pruned edge, and the
entire path survived the punning stage. Similarly, the point
$q_{t-1}$ is at most $\beta\delta(p,q)$ away from $q$ and again
for the right choice of $\beta$ every edge on the path that
approximates the distance between $q_{t-1}$ and $q$ survived the
punning stage. Thus, we get that there is a path whose length is
at most
$$(1+\eps)(3\beta\delta(p,q)+2\beta^2\delta(p,q))+\delta(p,q)+\beta\delta(p,q)~,$$
which is less than $(1+\eps)\delta(p,q)$ for $\beta<\eps/11$.

The second scenario is when $r'(q_{t-1}) \geq M_{i+1}$. In this
case, there is an edge in $H'$ from some $x \in
\Gamma_{i}(\NN(p,i))$ to some $y\in \Gamma_{i}(\NN(q_{t-1},i))$
whose length is at most $\delta(p,q)+\beta\delta(p,q)$, which is
not being pruned. We can build a path from $p$ to $q$ by
concatenating three segments as follows: A path from $p$ to $x$,
the edge $(x,y)$ and a path from $y$ to $q$. As before, for the
right choice of $\beta$ the paths from $p$ to $x$ and from $y$ to
$q$ are composed from edges that are shorter from $\delta(p,q)$,
the length of the shortest pruned edge, hence, from the minimality
$\delta(p,q)$ every edge on these paths survived the punning
stage. We get that there is a path whose length is at most
$$(1+\eps)(4\beta\delta(p,q)+5\beta^2\delta(p,q))+\delta(p,q)+
\beta\delta(p,q)~,$$
which is less than $(1+\eps)\delta(p,q)$ for $\beta<\eps/19$. This
completes the proof for case a.

{\bf Case b:} There is no point $q_t$ such that $\delta(q_t,q)> \beta
\delta(p,q)$.
In this case, let $q_{t-1}$ be the last point in the sequence of
points $q_1,\ldots q_i, \ldots$, where $q_i = s(q_{i-1})$ and $q_1
= q$. Similarly to before, there are two possible scenarios for
the spanner $H'$. Let $i=\ell(p,q_{t-1})$. The first scenario is
when $r'(q_{t-1}) < M_{i+1}$. In this case, there is an edge in
$H'$ from some $x \in \Gamma_{i}(\NN(p,i))$ to $q_{t-1}$ whose
length is at most $\delta(p,q)+\beta\delta(p,q)$. This edge could
not be pruned, since if it was pruned then $q_{t-1}$ could not
have been the last point in the sequence. Hence we can construct a
path from $p$ to $q$ exactly as we have done in the first scenario
of case a, described above. The second scenario is when
$r'(q_{t-1}) \geq M_{i+1}$. In this case, we can construct a path
from $p$ to $q$ exactly as we have done in the second scenario of
case a, described above.

This completes the proof of the induction base. The proof of the
general inductive step is almost identical. The only difference is
that when a path is constructed from $p$ to $q$, its portions from
$p$ to $x$ and from $q_{t-1}$ to $q$ in the first scenario and
from $p$ to $x$ and from $y$ to $q$ in the second scenario exist
in $\hat{H}$ by the induction hypothesis and not by the minimality
of $\delta(p,q)$.
\end{proof}

We end this section by stating its main Theorem. The proof of this
Theorem stems from Lemma~\ref{L-smaller} and
Lemma~\ref{L-Stretch1}.

\begin{theorem}
Let $(V,\delta)$ be a metric space of constant doubling dimension
with a radius assignment $r(\cdot)$ for its points and let
$I=(V,E,r)$ be its corresponding disk graph. Let
$I'=(V,E',r_{1+\eps})$ be the corresponding disk graph in the
relaxed model. It is possible to create a $(1+\eps)$-spanner of
size $O(n/\eps^d)$ for $I$ using edges of $I'$.
\end{theorem}

\section{Concluding remarks and open problems}
\label{s:con}

This paper presents a spanner construction for disk graphs in a
slightly relaxed model that is as good as spanners for complete
graphs and unit disk graphs. This result opens many other research
directions for disk graphs. We list here two questions that we
find particularly intriguing: Is it possible to design an
efficient compact routing scheme for disk graphs? And is it
possible to build an efficient distance oracle for disk graphs?

\bibliographystyle{plain}
\bibliography{peleg}

\end{document}

%% file: example-combined-tiny_m.pstex_t
\begin{picture}(0,0)%
\special{psfile=example-combined-tiny.pstex}%
\end{picture}%
\setlength{\unitlength}{3947sp}%
\begingroup\makeatletter\ifx\SetFigFont\undefined
\def\x#1#2#3#4#5#6#7\relax{\def\x{#1#2#3#4#5#6}}%
\expandafter\x\fmtname xxxxxx\relax \def\y{splain}%
\ifx\x\y   
\gdef\SetFigFont#1#2#3{%
  \ifnum #1<17\tiny\else \ifnum #1<20\small\else
  \ifnum #1<24\normalsize\else \ifnum #1<29\large\else
  \ifnum #1<34\Large\else \ifnum #1<41\LARGE\else
     \huge\fi\fi\fi\fi\fi\fi
  \csname #3\endcsname}%
\else
\gdef\SetFigFont#1#2#3{\begingroup
  \count@#1\relax \ifnum 25<\count@\count@25\fi
  \def\x{\endgroup\@setsize\SetFigFont{#2pt}}%
  \expandafter\x
    \csname \romannumeral\the\count@ pt\expandafter\endcsname
    \csname @\romannumeral\the\count@ pt\endcsname
  \csname #3\endcsname}%
\fi
\fi\endgroup
\begin{picture}(4577,2152)(226,-1577)
\put(1151,-100){\makebox(0,0)[lb]{\smash{\SetFigFont{14}{16.8}{rm}$y_1$}}}
\put(2147,-280){\makebox(0,0)[lb]{\smash{\SetFigFont{14}{16.8}{rm}$x_n$}}}
\put(150,-267){\makebox(0,0)[lb]{\smash{\SetFigFont{14}{16.8}{rm}$x_1$}}}
\put(150,-472){\makebox(0,0)[lb]{\smash{\SetFigFont{14}{16.8}{rm}$x_2$}}}
\put(1054,-1250){\makebox(0,0)[lb]{\smash{\SetFigFont{14}{16.8}{rm}$x_i$}}}
\put(3099, 79){\makebox(0,0)[lb]{\smash{\SetFigFont{14}{16.8}{rm}$y_n$}}}
\put(3293,-1139){\makebox(0,0)[lb]{\smash{\SetFigFont{14}{16.8}{rm}$X$}}}
\put(4073,-311){\makebox(0,0)[lb]{\smash{\SetFigFont{14}{16.8}{rm}$y_3$}}}
\put(4608,-701){\makebox(0,0)[lb]{\smash{\SetFigFont{14}{16.8}{rm}$y_2$}}}
\put(4803,-944){\makebox(0,0)[lb]{\smash{\SetFigFont{14}{16.8}{rm}$y_1$}}}
\end{picture}

%% file: proof-small_m.pstex_t
\begin{picture}(0,0)%
\special{psfile=proof-small.pstex}%
\end{picture}%
\setlength{\unitlength}{3947sp}%
\begingroup\makeatletter\ifx\SetFigFont\undefined
\def\x#1#2#3#4#5#6#7\relax{\def\x{#1#2#3#4#5#6}}%
\expandafter\x\fmtname xxxxxx\relax \def\y{splain}%
\ifx\x\y   
\gdef\SetFigFont#1#2#3{%
  \ifnum #1<17\tiny\else \ifnum #1<20\small\else
  \ifnum #1<24\normalsize\else \ifnum #1<29\large\else
  \ifnum #1<34\Large\else \ifnum #1<41\LARGE\else
     \huge\fi\fi\fi\fi\fi\fi
  \csname #3\endcsname}%
\else
\gdef\SetFigFont#1#2#3{\begingroup
  \count@#1\relax \ifnum 25<\count@\count@25\fi
  \def\x{\endgroup\@setsize\SetFigFont{#2pt}}%
  \expandafter\x
    \csname \romannumeral\the\count@ pt\expandafter\endcsname
    \csname @\romannumeral\the\count@ pt\endcsname
  \csname #3\endcsname}%
\fi
\fi\endgroup
\begin{picture}(3872,3484)(467,-2675)
\put(350,-871){\makebox(0,0)[lb]{\smash{\SetFigFont{14}{16.8}{rm}$p$}}}
\put(3500,-871){\makebox(0,0)[lb]{\smash{\SetFigFont{14}{16.8}{rm}$q$}}}
\put(3409,-1311){\makebox(0,0)[lb]{\smash{\SetFigFont{14}{16.8}{rm}$\beta\delta(p,q)$}}}
\put(4030,-126){\makebox(0,0)[lb]{\smash{\SetFigFont{14}{16.8}{rm}$>\beta/2\delta(p,q)$}}}
\put(3831,-583){\makebox(0,0)[lb]{\smash{\SetFigFont{14}{16.8}{rm}$q_{t-1}$}}}
\put(3831,300){\makebox(0,0)[lb]{\smash{\SetFigFont{14}{16.8}{rm}$q_{t}$}}}
\end{picture}